%% file: main.tex
\documentclass{llncs}
\usepackage{graphicx}
\usepackage{fullpage}
\usepackage{color}
\usepackage{colortbl}
\usepackage{amsmath}
\usepackage{amssymb}
\usepackage{url}
\usepackage{comment}
\usepackage[ruled,vlined]{algorithm2e}
\usepackage{hyperref}

\newcommand{\remove}[1]{{}}

\newcommand\NP{\textrm{NP}}

\newcommand\msize[1]{{\left|{#1}\right|}}
\newtheorem{observation}[theorem]{Observation}

\newcommand{\Height}{\textsf{M{\scriptsize IN}H{\scriptsize EIGHT}}}
\newcommand{\Sum}{\textsf{M{\scriptsize IN}S{\scriptsize UM}CW}}

\let\epsilon=\varepsilon

\makeatletter
\long\def\@caption#1[#2]#3{\par\addcontentsline{\csname
  ext@#1\endcsname}{#1}{\protect\numberline{\csname 
  the#1\endcsname}{\ignorespaces #2}}\begingroup
    \@parboxrestore
    \small
    \@makecaption{\csname fnum@#1\endcsname}{\ignorespaces #3}\par
  \endgroup}
\makeatother

\begin{document}
\title{Folding a Paper Strip to Minimize Thickness\thanks{
This research was performed in part at the 29th Bellairs Winter Workshop on
Computational Geometry. }
\thanks{
Erik Demaine was supported in part by NSF ODISSEI grant EFRI-1240383 and
NSF Expedition grant CCF-1138967. 
David Eppstein was supported in part by 
NSF grant 1228639 and ONR grant N00014-08-1-1015.
Adam Hesterberg was supported in part by DoD, Air Force Office of Scientific Research, National Defense Science and Engineering Graduate (NDSEG) Fellowship, 32 CFR 168a.
Hiro Ito was supported in part by 
JSPS KAKENHI Grant Number 24650006 and 
MEXT KAKENHI Grant Number 24106003. 
Anna Lubiw was supported in part by NSERC.
Ryuhei Uehara was supported in part by JSPS KAKENHI Grant Number 23500013 and 
MEXT KAKENHI Grant Number 24106004.
Yushi Uno was supported in part by KAKENHI Grant numbers 23500022 and 25106508.}}

\author{
Erik D. Demaine\inst{1}
\and
David Eppstein\inst{2}
\and
Adam Hesterberg\inst{3}
\and
Hiro Ito\inst{4}
\and\\
Anna Lubiw\inst{5}
\and
Ryuhei Uehara\inst{6}
\and
Yushi Uno\inst{7}
}

\institute{
Computer Science and Artificial Intelligence Lab, Massachusetts Institute of Technology, 
USA. 
\email{edemaine@mit.edu}
\and
Computer Science Department, University of California, Irvine, USA. 
\email{eppstein@uci.edu}
\and
Department of Mathematics, Massachusetts Institute of Technology, 
USA. 
\email{achester@mit.edu}
\and 
School of Informatics and Engineering, University of Electro-Communications, Japan. 
\email{itohiro@uec.ac.jp}
\and
David R. Cheriton School of Computer Science, University of Waterloo, Canada. 
\email{alubiw@uwaterloo.ca}
\and 
School of Information Science, Japan Advanced Institute of Science and Technology, 
Japan. 
\email{uehara@jaist.ac.jp}
\and
Graduate School of Science, Osaka Prefecture University, Japan. 
\email{uno@mi.s.osakafu-u.ac.jp}
}
\maketitle

\begin{abstract}
In this paper, we study how to fold a specified origami crease pattern in
order to minimize the impact of paper thickness.  Specifically,
origami designs are often expressed by a mountain-valley pattern
(plane graph of creases with relative fold orientations), but in general
this specification is consistent with exponentially many possible folded states.
We analyze the complexity of finding the best consistent folded state
according to two metrics: minimizing the total number of layers in the
folded state (so that a ``flat folding'' is indeed close to flat), and
minimizing the total amount of paper required to execute the folding
(where ``thicker'' creases consume more paper).
We prove both problems strongly \NP-complete even for 1D folding.
On the other hand, we prove the first problem fixed-parameter tractable
in 1D with respect to the number of layers.
\end{abstract}

\setcounter{footnote}{0}

\section{Introduction}

Most results in computational origami design assume an idealized,
zero-thickness piece of paper.  This approach has been highly successful,
revolutionizing artistic origami over the past few decades.  Surprisingly
complex origami designs are possible to fold with real paper thanks in part to
thin and strong paper (such as made by Origamido Studio) and perhaps also to some
unstated and unproved properties of existing design algorithms.

This paper is one of the few attempts to model and optimize the effect of
positive paper thickness.  Specifically, we consider an origami design
specified by a \emph{mountain-valley pattern}
(a crease pattern plus a mountain-or-valley assignment for each crease),
which in practice is a common specification for complex origami designs.
Such patterns only partly specify a folded state, which also consists of
an \emph{overlap order} among regions of paper.
In general, there can be exponentially many overlap orders consistent with
a given mountain-valley pattern.
Furthermore, it is \NP-hard to decide flat foldability of a mountain-valley
pattern, or to find a valid flat folded state (overlap order) given the promise
of flat foldability \cite{Bern-Hayes-1996}.
But for 1D pieces of paper, the same problems are polynomially solvable
\cite{Arkin-Bender-Demaine-Demaine-Mitchell-Sethia-Skiena-2003,DemaineORourke2007},
opening the door for optimizing the effects of paper thickness among the
exponentially many possible flat folded states---the topic of this paper.

\subsubsection{Preceding Research}

One of the first mathematical studies about paper thickness is also primarily about
1D paper.  Britney Gallivan \cite{Gallivan}, as a high school student, modeled and analyzed the
effect of repeatedly folding a positive-thickness piece of paper in half.
Specifically, she observed that creases consume a length of paper
proportional to the number of layers they must ``wrap around'',
and thereby computed the total length of paper (relative to the paper
thickness) required to fold in half $n$ times.
She then set the world record by folding a 4000-foot-long piece of (toilet)
paper in half twelve times, experimentally confirming her model and analysis.

Motivated by Gallivan's model,
Uehara~\cite{Uehara2010b} defined the \emph{stretch} at a crease to be the
number of layers of paper in the folded state that lie between the two paper
segments hinged at the crease.
We will follow the terminology of
Umesato et al.~\cite{UmesatoSaitohUeharaItoOkamoto2013} who later replaced the
term ``stretch'' with \emph{crease width}, which we adopt here.
Both papers considered the case of a strip of paper with \emph{equally spaced}
creases but an arbitrary mountain-valley assignment.
When the mountain-valley assignment is uniformly random, its expected number of consistent folded states is $\Theta(1.65^n)$ \cite{Uehara2011b}.
Uehara~\cite{Uehara2010b} asked whether it is NP-hard, for a given
mountain-valley assignment, to minimize the maximum crease width or
to minimize the total crease width (summed over all creases).
Umesato et al.~\cite{UmesatoSaitohUeharaItoOkamoto2013} showed that the
first problem is indeed NP-hard, while
the second problem is fixed-parameter tractable.   
Also, there is a related study for a different model, 
which tries to compact orthogonal graph drawings to use minimum number 
of rows \cite{BannisterEppsteinSimons2012}. 

\subsubsection{Models}

We consider the problem of minimizing crease width in the more general
situation where the creases are not equally spaced along the strip of paper.  
This more general case has some significant differences with the equally spaced case.
For one thing, if the creases are equally spaced, all mountain-valley patterns
can be folded flat by repeatedly folding from the rightmost end;
in contrast, in the general case, some mountain-valley patterns (and even some
crease patterns) have no consistent flat folded state
that avoids self-intersection.
Flat foldability of a mountain-valley pattern can be checked in linear time
\cite{Arkin-Bender-Demaine-Demaine-Mitchell-Sethia-Skiena-2003} \cite[Sec.~12.1]{DemaineORourke2007}, but it requires a nontrivial algorithm.

For creases that are not equally spaced, the notion of crease width must also
be defined more precisely, because it is not so clear how to count the layers
of paper between two segments at a crease. 
For example, in Fig.~\ref{fig:cw}, 
although no layers of paper come all the way to touch the three creases
on the left, we want the sum of their crease widths to be $100$. 

\begin{figure}[htb]
  \centering
  \includegraphics[width=0.6\textwidth]{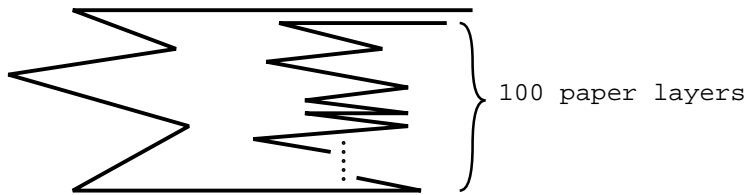}
  \caption{How can we count the paper layers?}
  \label{fig:cw}
\end{figure}

We consider a folded state to be an assignment of the segments to horizontal
\emph{levels} at integer $y$ coordinates, with the creases becoming vertical
segments of variable lengths.
See Fig.~\ref{fig:models} and the formal definition below.
Then the \emph{crease width} at a crease is simply the number of levels 
in between the levels of the two segments of paper joined by the crease. That is,
it is one less than the length of the vertical segment assigned to the crease.
This definition naturally generalizes the previous definition for
equally spaced creases.
Analogous to Uehara's open problems~\cite{Uehara2010b},
we will study the problems of minimizing the maximum crease width and
minimizing the total crease width for a given mountain-valley pattern.
The total crease width corresponds to the extra length of paper needed to fold
the paper strip using paper of positive thickness, naturally generalizing
Gallivan's work\footnote{
Although we assume orthogonal bends
  in this paper, while Gallivan measures turns as circular arcs, this changes the length by only a constant factor. Gallivan's model seems
  to correspond better to practice.}
\cite{Gallivan}.

\begin{figure}[htb]
  \centering
  \includegraphics[width=0.9\textwidth]{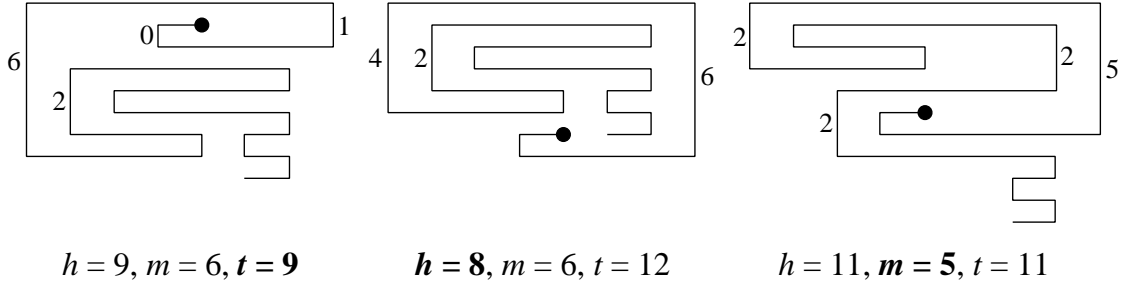}
\caption{Three different folded states of the crease pattern $VMVMVVMMMM$ (ending at the dot).
The crease width of each crease is given beside its corresponding vertical segment.
Each folding is better than the other two in one of the three measures, where $h$ is the height, $m$ is the maximum crease width, and $t$ is the total crease width.  
}
\label{fig:models}
\end{figure}

In the setting where creases need not be equally spaced, there is another
sensible measure of thickness: the \emph{height} of the folded state is
the total number of levels.  The height is always $n+1$ for $n$ equally spaced
creases, but in our setting different folds of the same crease pattern can have different heights.  
Figure~\ref{fig:models} shows how the
three measures can differ.  Of course, the maximum crease width is always
less than the height.

\subsubsection{Contributions}

Our main results (Section~\ref{NP-completeness})
are \NP-hardness of the problem of minimizing height
and the problem of minimizing the total crease width.
See Table~\ref{tbl:variations}.
In addition, we show in Section~\ref{Fixed-parameter tractability}
that the problem of minimizing height is fixed-parameter
tractable, by giving a dynamic programming algorithm that runs in
$O(2^{O(k \log k)}n)$ time, where $k$ is the minimum height. 
This dynamic program can be adapted to minimize maximum crease width
or total crease width for foldings of bounded height,
with the same time complexity as measured in terms of the height bound.
Table~\ref{tbl:variations} summarizes related results. 

\begin{table}[ht]
\centering
\tabcolsep=0.5em
\def\arraystretch{1.2}
\begin{tabular}{c|cc}\hline
\textbf{thickness measure}&\textbf{eq. spaced creases}&\textbf{~~general creases~~}\\
\hline
~~height~~&trivial&\NP-hard (this paper)\\
& & FPT wrt. min height (this paper)\\
~~max crease width~~&\NP-hard~\cite{UmesatoSaitohUeharaItoOkamoto2013} &
\llap{$\implies$\qquad} \NP-hard~\cite{UmesatoSaitohUeharaItoOkamoto2013}\\
~~total crease width~~&open&\NP-hard (this paper)\\ \hline
\end{tabular}
\medskip
\caption{Complexity of minimizing thickness, by model, for the case of equally spaced creases and for the general case.}
\label{tbl:variations}
\end{table}


\remove{
They introduced a 
notion of ``crease width at a crease''
and considered the problem of minimizing the maximum crease width 
for a given paper strip with creases at regular intervals.
They consider only the case when the creases are equally spaced,
which puts every crease at either the left end or the right end of the folded state,
and they focus on the local minimum at each crease in a folded state.
The problem is fundamentally different if the creases are allowed to be in arbitrary positions.
For example, if they are equally spaced, any mountain-valley pattern can be folded flat 
by repeatedly folding from right-most end, but in the general case, some patterns cannot be 
folded flat without the paper intersecting itself.
Flat foldability of a paper strip can be checked in linear time \cite[Sec.~12.1]{DemaineORourke2007}, 
but it is not trivial.

If the paper strip can be folded flat with a given mountain-valley pattern,
how can we estimate the thickness of a folded state?
Also, how can we find the thinnest folded state?
We consider this new problem in this paper.
Unlike the equally-spaced case, 
we have to consider the global structure of a folded state,
and hence we meet different problems.
In this paper, we focus on this global minimum thickness problem.

We first introduce two notions, of the ``thickness'' and ``height'' of each point in a folded state.
The thickness of a point is the number of paper layers at that point.
When we fix a mountain-valley pattern, the thickness is invariant 
and can be computed in linear time. 
On the other hand, the height will be a realistic measure of thickness of the folded state,
dependent on the folding, and we can look for the best foldings with 
respect to that measure of height.
First we show that the problem of minimizing the maximum height of 
a paper strip for a given mountain-valley pattern of length $n$ is strongly \NP-complete in general.
Hence we cannot solve the problem in polynomial time unless P=NP.

In \cite{UmesatoSaitohUeharaItoOkamoto2013}, Umesato et al.~defined 
the notion of crease width at a crease
by the number of paper layers at that crease,
They considered the maximum and total crease width for all creases
and showed 
that minimizing maximum crease width was \NP-complete,
but left open the computational complexity of minimizing total crease width.
We first mention that the notion of crease width at a crease depends on the fact that their paper strip
 has creases at regular intervals and all creases are piled at 
the two endpoints of a folded state, so the number of paper layers at the crease is a reasonable measure of its size.
In general case, we cannot use that notion.
For example, in Fig.~\ref{fig:cw}, 
we have to put a hundred paper layers between two long paper segments,
and thus they are put between the three leftmost creases.
In the original notion of the crease width, these three creases have crease width 0
since there are no paper layers at each of them.
However, we'll define crease width in such a way that the
total crease width at the left side of this folded state is 100,
using the paper layers between each pair of the paper segments of the paper strip.
This is a natural extension of the original notion of total crease width,
and this is also similar to an idea of Gallivan's to estimate the length of toilet paper 
needed to fold it in half twelve times \cite{Gallivan}.
The reduction in the proof of our main theorem can be used to show the \NP-completeness 
of this total crease width minimization problem, 
answering the open question in \cite{UmesatoSaitohUeharaItoOkamoto2013}.

Then we propose a nontrivial dynamic programming algorithm that solves the minimization problem.
It runs in $O(2^{O(h \log h)}n)$ time, where $h$ is the maximum height of the folded state.
Thus the problem is fixed parameter tractable with respect to the maximum height.
}

\section{Preliminaries}
\label{sec:preliminaries}

We model a {\em paper strip} as a one-dimensional line segment.
It is rigid except at {\em creases} $p_1,p_2,\ldots,p_n$ on it;
that is, we are allowed to fold only at these crease points.
For notational convenience, the two ends of the paper strip are denoted by $p_0$ and $p_{n+1}$.
We are additionally given a {\em mountain-valley string} 
$s=s_1s_2\cdots s_n$ in $\{M,V\}^n$.
In the \emph{initial state} the paper strip is placed on the $x$-axis, with each crease $p_i$ at a given coordinate 
$x_i$. 
Without loss of generality, we assume that 
$x_0 = 0 < x_1 < \cdots < x_n < x_{n+1}$. 
Sometimes we will normalize so $x_{n+1} = 1$. 
We may consider the paper strip as a sequence of $n+1$ 
{\em segments} $S_i$ of length $x_{i+1} - x_i$ 
delimited by the creases $p_i$ and $p_{i+1}$ for each $i \in \{0,1,\ldots,n\}$.
We fold the strip through two dimensions, so we distinguish the {\it top} side of the strip (the positive $y$ side) and the {\it bottom} side of the strip (the negative $y$ side).
Each crease's letter determines how we can fold it:
when it is $M$ (resp. $V$), the two paper segments sharing the crease are folded in 
the direction such that their bottom sides (resp. top sides) are close to touching 
(although they may not necessarily touch if they have other paper layers between them).

Following Demaine and O'Rourke~\cite{DemaineORourke2007}
we define a \emph{flat folding} (or \emph{folded state})
via the relative stacking order of collocated layers of paper.
We begin with $x_0$ at the origin, and the first segment lying in the positive $x$-axis. 
The lengths of the segments determine where each segment lies along the $x$-axis (because they zig-zag).  Suppose that point $p_i$ is mapped to $x$-coordinate $f(p_i)$.  
The mountain-valley assignment determines for each segment $S_i$ whether $S_i$ lies above or below $S_{i+1}$. 
We extend this to specify the 
relative vertical order of any two segments that overlap horizontally.
This defines a \emph{folded state} so long as the vertical ordering of segments is transitive and \emph{non-crossing}. More formally:
\begin{enumerate}
\item if segments $S_i$ and $S_{i+1}$ are joined by a crease at $x$-coordinate $f(p_i)$  then for any segment $S$ that extends to the left and the right of $f(p_i)$, either $S <S_i, S_{i+1}$ or $S > S_i, S_{i+1}$, 
\item if segments $S_i$ and $S_{i+1}$ are joined by a crease at $x$-coordinate $f(p_i)$, segments $S_j$ and $S_{j+1}$ are joined by a crease at the same $x$-coordinate $f(p_j) = f(p_i)$, and all 4 segments extend to the same side of the crease, then the two creases do not \emph{interleave}, i.e., we do not have $A < B < A' < B'$ where $A$ and $A'$ are one of the pairs joined at a crease and $B$ and $B'$ are the other pair.
\end{enumerate}
 
When the $x_i$'s are not equally spaced, the paper strip cannot necessarily be folded
flat with the given mountain-valley assignment.  For example, segments of lengths 2, 1, 2 do not allow the assignment $VV$.
There is a linear time algorithm to test whether an assignment has a flat folding~\cite{DemaineORourke2007}.

In order to define crease width, we will use an enhanced notion of folded states:
a \emph{leveled folded state} is an assignment of the segments to \emph{levels} from the set $\{1, 2, \ldots \}$
such that the resulting vertical ordering of segments is a valid folded state.  See Fig.~\ref{fig:models}.
We can draw a leveled folded state as a rectilinear path of alternating horizontal and vertical segments, where the horizontal segments are the given ones, and the vertical segments (which represent the creases) have variable lengths.

Clearly a leveled folded state provides a folded state, but in the reverse direction, a folded state may correspond to many leveled folded states.  However, for the measures we are concerned with, we can efficiently compute the best leveled folded state corresponding to any folded state.

The \emph{height} of a leveled folded state is the number of levels used.  Given a folded state, the minimum height of any corresponding leveled folded state can be computed efficiently, since it is the length of a longest chain in the partial order defined on the segments in the folded state.

The \emph{crease width} of a crease in a leveled folded state is the number of levels in-between the two segments joined at the crease.  We are interested in minimizing the maximum crease width and in minimizing the total crease width, i.e., the sum of the crease widths of all the creases. 
In both cases, given a folded state, we can compute the best corresponding leveled folded state using linear programming.

\remove{
\begin{figure}
  \centering
  \includegraphics[width=0.6\textwidth]{define}
  \caption{Height and thickness of a folded paper strip.}
  \label{fig:define}
\end{figure}

Now we introduce a new notion of the {\em height} and {\em thickness} of a folded state of the paper strip.
Let the folded paper strip be placed on the $x$-axis.
Then, intuitively, each paper segment is supposed to be flat and rigid in a folded state.
Therefore, if the segment is supported by some different segments from below,
the maximum of their heights determines its {\em height}.
On the other hand, we usually consider the {\em thickness} of a layered paper strip
to be the maximum number of paper layers at a point.
For example, Fig.~\ref{fig:define} gives two folded states for a paper strip 
with the same mountain-valley assignment.
It is easy to see that the folded states have the same maximum thickness, 5,
regardless of how they are folded, but the heights depend on it:
one folded state (a) has height 6 and the other one (b) has height 8.
Precisely, we imagine that we first place the folded paper strip on the $x$-axis.
Then, some paper segments directly touch the $x$-axis; we define them to have height 0.
Each other paper segment $p_ip_{i+1}$ is on at least one 
other paper segment $p_{j}p_{j+1}$. Then the {\em height} 
$h(p_ip_{i+1})$ of $p_ip_{i+1}$ is defined 
by $\max\{ h(p_jp_{j+1})\}+1$, where $j$ ranges over all paper segments $p_j p_{j+1}$ 
directly touching $p_i p_{i+1}$ at some point between $p_i$ and $p_{i+1}$.
Note that these segments are considered as open intervals. 
That is, if $p_j p_{j+1}$ touches $p_{i} p_{i+1}$ only at their ends, 
they are independent.
For a mountain-valley string $s$, we call a folded state {\em legal with respect to $s$} if it follows the string.
}

A mountain-valley string that alternates $MVMVMV\ldots$ is called a {\em pleat}.
For equally-spaced creases, the legal folded state is unique (up to reversal of the paper) 
if and only if $s$ is a pleat \cite{Uehara2010b,Uehara2011b}.

In this paper, we consider three versions of minimizing thickness 
in a flat folding.  
For all three problems we have the following instance in common: 

\smallskip
\noindent
INSTANCE: 
A paper strip $P$, with creases $p_1,\ldots,p_n$ at positions $x_1, \ldots, x_n$
with a mountain-valley string $s \in \{M, V\}^n$, and a natural number $k$. 

\medskip
The questions of the three problems are as follows:

\medskip
\noindent
{\bf\textsf{MinHeight}}: Is there a leveled folded state of height at most $k$?

\smallskip
\noindent
{\bf\textsf{MinMaxCW}}: Is there a leveled folded state with maximum crease width at most $k$?

\smallskip
\noindent
{\bf\textsf{MinSumCW}}: Is there a leveled folded state with total crease width at most $k$?


\remove{
In this paper, we consider the\Height{} problem of minimizing the maximum height of a folded state, as follows:

\smallskip\noindent
{\bf INSTANCE}: A paper strip $P$, creases $p_1,\ldots,p_n$ on $P$
with a mountain-valley string $s \in \{M, V\}^n$, and a natural number $k$.\\
{\bf QUESTION}: Is there a legal folded state with respect to $s$ of maximum height at most $k$?
\smallskip

Given a point $x$ on the $x$-axis in a folded state, 
the layers of segments at $x$ can be represented by 
the ordering of the segments in the increasing order with respect to $y$;
for example, the lefts part of thickness 5 in Fig.~\ref{fig:define}(a) 
or (b)
are $[S_0|S_1|S_2|S_3|S_4]$,
the right part of thickness 5 in Fig.~\ref{fig:define}(a) is 
$[S_7|S_8|S_4|S_5|S_6]$, and
the right part of thickness 5 in Fig.~\ref{fig:define}(b) is $[S_4|S_5|S_7|S_8|S_6]$.
}

\section{\NP-completeness}
\label{NP-completeness}

In this section, we show \NP-completeness of 
the {\sf MinHeight} and {\sf MinSumCW} problems. 
We remind the reader that the pleat folding has  a unique folded state \cite{Uehara2010b,Uehara2011b}.
We borrow some useful ideas from \cite{UmesatoSaitohUeharaItoOkamoto2013}.

\begin{figure}
   \centering
   \includegraphics[width=0.48\textwidth]{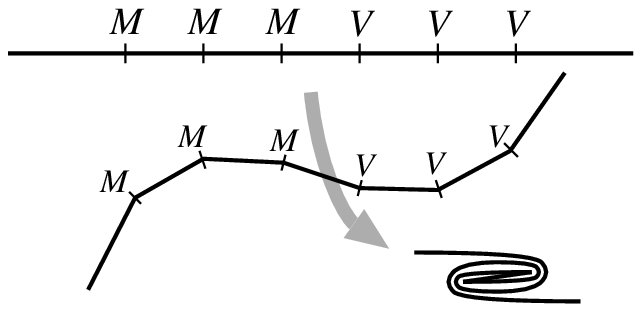}
   \caption{The unique flat folding of the string MMMVVV.}
   \label{fig:mmm}
\end{figure}


\begin{observation}
\label{obs:spiral}
Let $n$ be a positive integer,
$P$ be a strip with creases $p_1,\ldots,p_{2n}$,
and $s$ be a mountain-valley string $M^nV^n$.
We suppose that the paper segments are of equal length except a longer one at each end.
Precisely, we have $\msize{S_i}=\msize{S_j} < \msize{S_0} = \msize{S_{2n}}$ for all $i,j$ with $0<i,j<2n$,
where $\msize{S_i}$ denotes the length of the segment $S_i$,
Then the legal folded state with respect to $s$ is unique up to reversal of the paper.
Precisely, the legal folded state
has the segments in vertical order  $S_0,S_{2n-1},S_2,S_{2n-3}, \dots,S_{2i},S_{2(n-i)-1},\dots,S_1,S_{2n}$ or the reverse.
\end{observation}
%
A simple example is given in Fig.~\ref{fig:mmm}.
We call this unique folded state the {\em spiral folding} of size $2n$.  

Our hardness proofs reduce from {\sf 3-PARTITION}, defined as follows. 

\begin{quote}
{\sf 3-PARTITION} (cf.~\cite{GJ79}) \\
Instance: 
A finite multiset $A=\{a_1,a_2,\ldots,a_{3m}\}$ of $3m$ positive integers.
Define $B = \sum_{j=1}^{3m} a_j / m$.
We may assume each $a_j$ satisfies $B/4 < a_j < B/2$. \\
Question: 
Can $A$ be partitioned into $m$ disjoint sets $A^{(1)}$, $A^{(2)}$, 
$\dots$, $A^{(m)}$ 
such that $\sum A^{(i)} =B$ for every $i$ with $1 \le i \le m$?
\end{quote}

%
%

\noindent
It is well-known that {\sf 3-PARTITION} is strongly \NP-complete, i.e., 
it is \NP-hard even if the input is written in unary notation \cite{GJ79}.
Our reductions are based on a similar reduction of
Umesato et al.~\cite{UmesatoSaitohUeharaItoOkamoto2013}.

\begin{theorem}
\label{th:NP}
The {\sf MinHeight} problem for paper folding height is \NP-complete.
\end{theorem}
\begin{proof}
It is easy to see that the problem is in \NP.
To prove hardness, we reduce from {\sf 3-PARTITION}. 

\begin{figure}
  \centering
  \includegraphics[width=\textwidth]{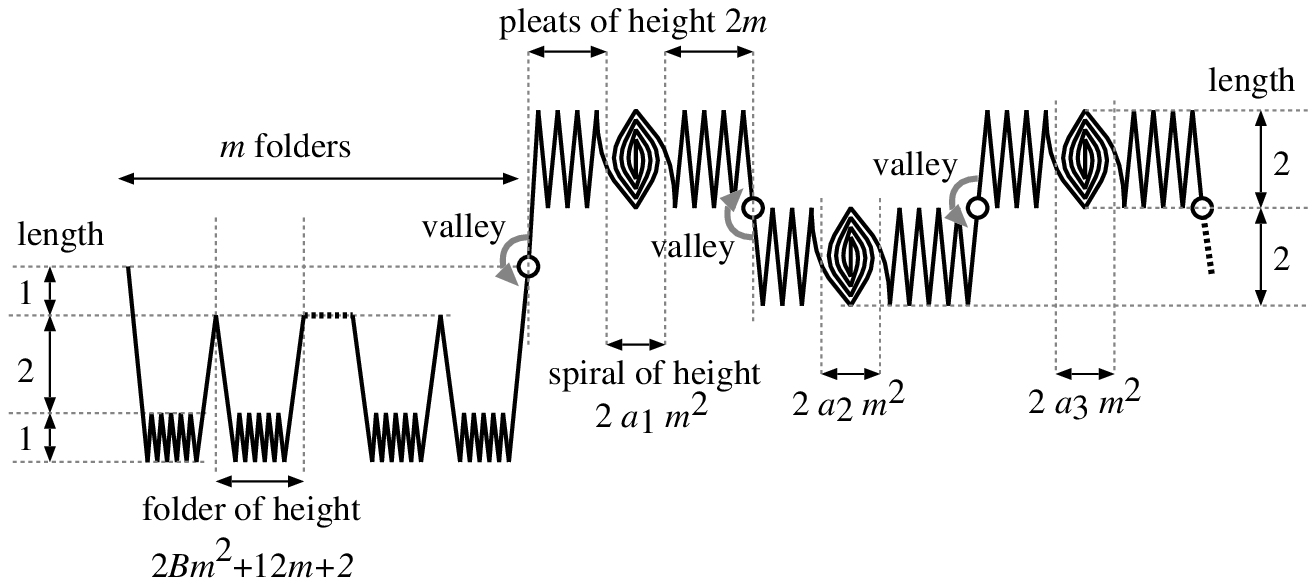}
  \caption{Outline of the reduction for Theorem~\ref{th:NP}.  Note that this figure and the next one are sideways compared to previous figures, so height is horizontal.}
  \label{fig:outline}
\end{figure}

Given an instance $\{a_1,a_2,\dots,a_{3m}\}$ of {\sf 3-PARTITION},
we construct a corresponding paper strip $P$ as follows
(Fig.~\ref{fig:outline}).
The left part of $P$ is folded into $m$ \emph{folders}, where
each folder is a pleat consisting of $2 B m^2 + 12 m$ \emph{short} segments of
length $1$ between two segments of length~$3$, except for the very first and
last long segments, which have length~$4$.%
\footnote{In the reduction in \cite{UmesatoSaitohUeharaItoOkamoto2013},
  this folder consists of just two segments.}
The right part of $P$ contains $3 m$ gadgets,
where the $i$th gadget represents the integer~$a_i$.
The $i$th gadget consists of 
one spiral of height $2 a_i m^2$ between two $2 m$ pleats.
Each line segment in the gadget has length $2$
except for the one end segment which has length $3$.
This construction can be carried out in polynomial time.

By Observation~\ref{obs:spiral}, each spiral folds uniquely, 
and also we know that each pleat folds uniquely 
\cite{Uehara2010b,Uehara2011b}. 
Therefore, the folders and gadgets fold uniquely.
Figure~\ref{fig:outline} shows the unique combination of these foldings
before folding at the \emph{joints}, depicted by white circles.
Once the joints are valley folded, the folding will no longer be unique.
%

The intuition is that the pleats of each gadget give us the freedom to place
the spiral of each gadget in any folder.  The heights of the spirals ensure
that the packing of spirals into folders acts like {\sf 3-PARTITION}.  
More precisely, we show:

\begin{claim}
\label{clm:reduction}
An instance $(A,B)$ of {\sf 3-PARTITION} has a solution if and only if
the paper strip $P$ can be folded with height at most $2 B m^3+12m^2+2m$.
\end{claim}

\begin{figure}
  \centering
  \includegraphics[width=0.8\textwidth]{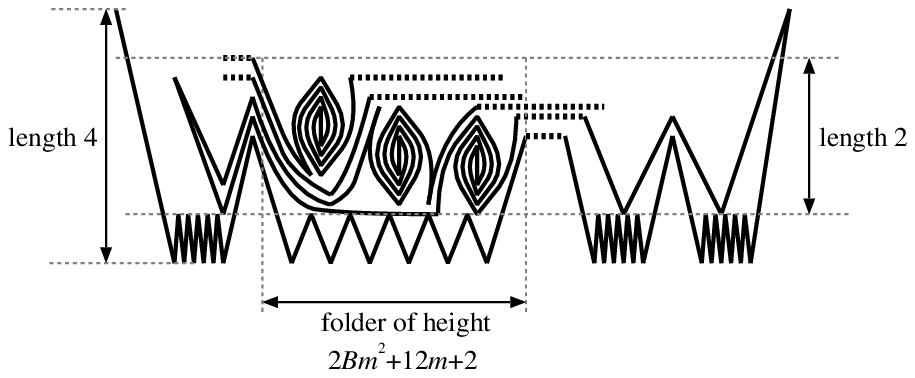}
  \caption{Putting spirals into a folder.}
  \label{fig:folder}
\end{figure}


To prove the claim,
first suppose that the {\sf 3-PARTITION} instance $\{a_1$, $a_2$, 
$\ldots$, $a_{3m}\}$
has a solution, say, $A^{(1)},A^{(2)},\dots,A^{(m)}$.
Then we have $A^{(i)} \subset A$, $\msize{A^{(i)}}=3$, 
$\sum A^{(i)} = B$ for each $i$ in $\{1,2,\dots,m\}$,
and $A=\dot{\bigcup}_{i=1}^{m} A^{(i)}$.
For the three items in $A^{(i)}$, we put the three corresponding spirals
into the $i$th folder; see Fig.~\ref{fig:folder}.
Because the items sum to $B$, the total height of the spirals is $2 B m^2$.
Each gadget uses $2(m-1)$ of the $4m$ total pleats to position its spiral,
leaving $2(m+1)$ pleats which we put in the folder of the spiral,
for a total of $6(m+1)$.
The $3m-3$ other gadgets also place two pleats in this spiral,
just passing through, for a total of $6m-6$.
Thus each folder has at most $2 B m^2 + 12 m$ layers added and,
because it already had $2 B m^2+12 m$ short pleat segments,
its final height is $2 B m^2+12 m+2$ (including the two long segments).
Therefore the total height of the folded state is $2 B m^3+12 m^2+2m$ as desired.

Next suppose that the paper strip $P$ can be folded with height at most
$k=2 B m^3+12m^2+2m$.
There are $m$ folders each with height at least $2 B m^2 + 12 m + 2$.
Therefore, each folder must have height exactly $2 B m^2 + 12 m + 2$
and the number of levels inside the folder is $2 B m^2 + 12 m$.
Furthermore, the spirals must be folded into the folders.
We claim that the spirals in each folder must have total height
at most $2 B m^2$. For,
if the spirals in one of the folders have total height more than $2 B m^2$, then
they have height at least $2(B+1) m^2 = 2 B m^2 + 2 m^2$,
which is greater than $2 B m^2 + 12 m$ if $2 m^2 > 12 m$, i.e., if $m > 6$
(which we may assume without loss of generality). 
In particular, each folder must have at most three spirals:
because each $a_j$ is greater than $B/4$, 
each spiral has height larger than $B m^2/2$, 
so four spirals would have height larger than $2 B m^2$.
Because the $3 m$ spirals are partitioned among $m$ folders,
exactly three spirals are placed in each folder,
and their total height of at most $2 B m^2$ corresponds to three elements
of sum at most (and thus exactly)~$B$.
Therefore we can construct a solution to the {\sf 3-PARTITION} instance.
\qed
\end{proof}

\remove{
\paragraph{Total Crease Width:}
As mentioned in the introduction, the notion of the original total crease width introduced 
in \cite{UmesatoSaitohUeharaItoOkamoto2013} does not measure 
the amount of paper needed for real folding when the creases are not equally spaced.
We extend it to count the number of paper layers for each pair of segments.
Precisely, the {\em total crease width} of a folded state is defined as follows.
First, we define a function $\delta(S_i,S_j)$ by $\delta(S_i,S_j)=1$ if they share 
a common point in the folded state, or $\delta(S_i,S_j)=0$ otherwise.
Recall that $h(S_i)$ is the height of $S_i$ in the folded state.
Now the total crease width of a folded state is defined by 
\[
\sum_{S_i,S_j}\delta(S_i,S_j)\left|h(S_i)-h(S_j)\right|.
\]
Intuitively, it is the total length of paper needed for all the creases.
This definition is a natural extension of the original one, and we define the \Sum{} problem,
with the same inputs as the \Height{} problem, but asking whether a total crease width of at most $k$
is achievable.\footnote{When we consider total crease width,
the value of $k$ is different from the value of the maximum height.
We omit the details, but they are not difficult to state.}
}

\begin{theorem}
The {\sf MinSumCW} problem is \NP-complete. 
\end{theorem}
\begin{proof}
This reduction from {\sf 3-PARTITION} 
is a modification to the reduction to {\sf MinHeight} 
in the proof of Theorem~\ref{th:NP}; 
refer to Figures~\ref{fig:outline2} and~\ref{fig:folder2}.
We introduce a deep ``molar'' at both ends of each gadget,
which must fit into deep ``gums'' at either end of the folders.
Specifically, for $z=m^4$,
each gum has $2z+4m$ pleats, and each molar
in the $i$th gadget has $2z+4(m-i)$ pleats.
In the intended folded state, the left molars nest inside each other
(smaller/later inside larger/earlier) within the left gum, and similarly
for the right molars into the right gum.
In this case, every molar and every gum remains
at its minimum possible height given by its pleats.

\begin{figure}[htb]
  \centering
  \includegraphics[angle=270,width=\textwidth]{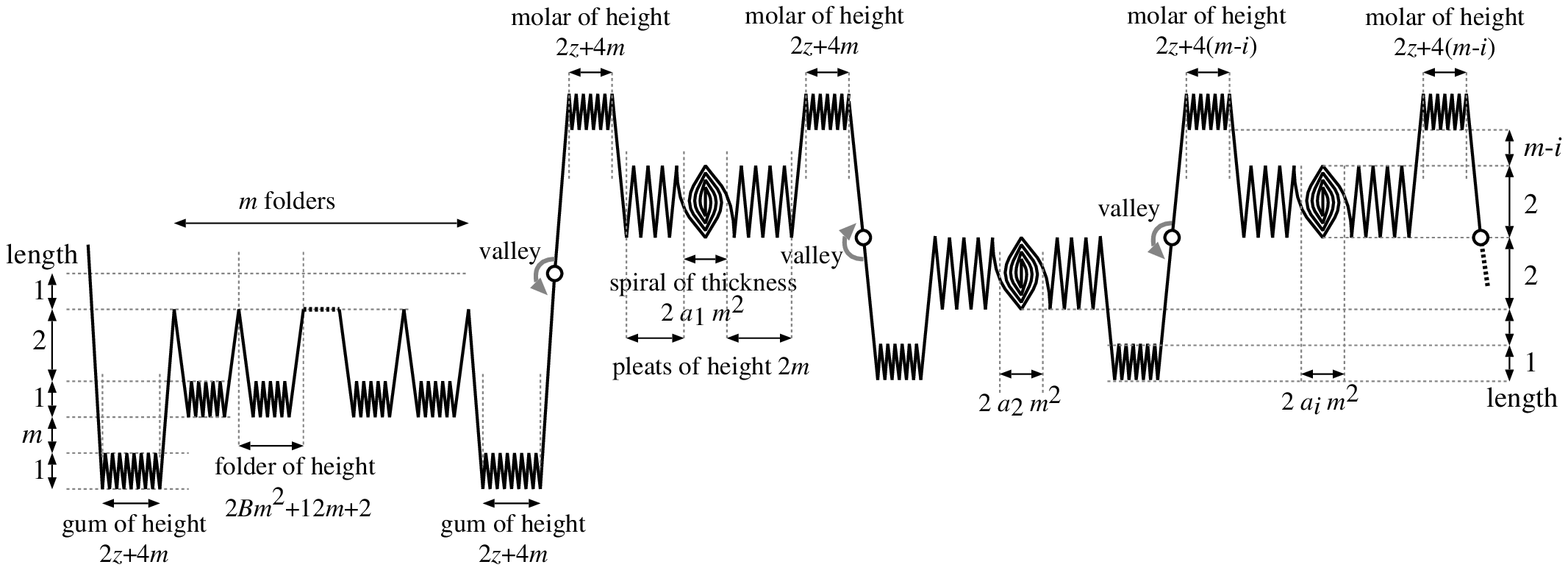}
  \caption{Outline of the reduction.  Note that height is vertical.}
  \label{fig:outline2}
\end{figure}

\begin{figure}
  \centering
  \includegraphics[width=\textwidth]{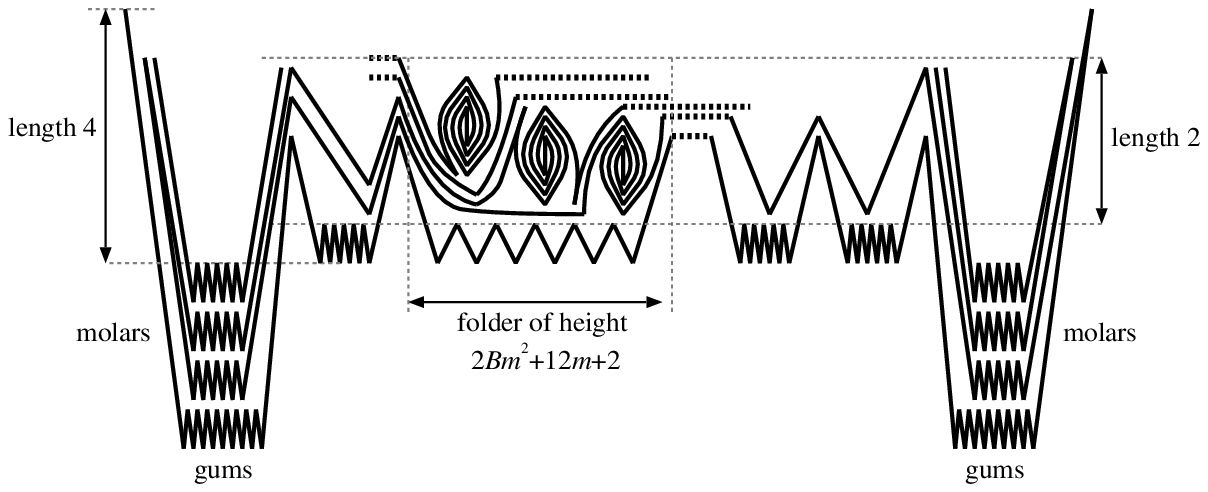}
  \caption{Putting spirals into folders and molars into gums.}
  \label{fig:folder2}
\end{figure}

The heights of the molars guarantee that, in any legal folding, every
molar ends up in a gum.
If, in any of the $m$ gadgets, the right molar folds into the
left gum, then the left molar of that gadget also folds into the
left gum, so the left gum has height at least $4 z$ in the folded state,
$2 z - 4 m$ more than its minimum height.
This increase in height translates into an equal increase in the total
crease width (because the number of creases remains fixed).
Because $z = m^4$, this increase will dominate the total crease width.
Therefore every folding with a right molar in the left gum,
or with a left molar in the right gum, has
total crease width larger than the intended folded state.

This argument guarantees that, in any solution folding to the \Sum{}
instance, each gadget has its left molar in the left gum and its right
molar in the right gum.  In this case, the height of each gadget is the
height of its spiral plus the height of all the folders, which will be
minimized precisely when the folders do not grow in height.
The total crease width of a gadget differs from its height by a fixed amount
(the number of creases), so we arrive at the same minimization problem.
Thus the proof reduces to the {\sf MinHeight} construction.
\qed
\end{proof}

\section{Fixed-parameter tractability}
\label{Fixed-parameter tractability}

In this section, we show the following theorem. 

\input{dynprog.tex}

\section{Conclusion}

In this paper, we considered three problems {\sf MinHeight}, {\sf MinMaxCW} 
and {\sf MinSumCW} for 1D strip folding, and showed some intractable results. 
We have some interesting open questions. 
Although we gave an FPT algorithm for {\sf MinHeight}, 
it is not clear if the other two problems have FPT algorithms. 
Extending our models to 2D foldings would also be interesting.

{\raggedright
\bibliographystyle{abbrv}
\bibliography{thick}}


\end{document}

%% file: dynprog.tex

\begin{theorem}
Testing whether a strip with $n$ folds has a folded state with height 
at most $k$ can be done in time $O(2^{O(k\log k)}n)$. 
\end{theorem}

\begin{proof}
We use a dynamic programming algorithm that sweeps from left to right across the line onto which the
 strip is folded, stopping at each of the points on the line where a strip endpoint or a crease (fold point) is placed. At each point of the line between two stopping points, there can be at most $k$
segments of the strip,
for otherwise the height would necessarily be larger than $k$ and we could terminate the algorithm,
 returning that the height is not less than or equal to $k$.
We define a \emph{level assignment} for a point $p$ between two stopping points
to be a function $a$ from input segments that overlap $p$ to distinct integer levels from $1$ to $k$. The number of possible level assignments for any point is therefore at most $k^{k}$.

Let $\epsilon>0$ be smaller than the distance between any two stopping points.
At each stopping point $p$ of the algorithm, we will have a set $A$ of allowed level assignments $a^-$ for the point $p-\epsilon$; initially (for the leftmost point of the folded input strip) $A$ will contain the unique level assignment for the empty set of segments.
For each combination of a level assignment $a^-$ in $A$ for the point $p-\epsilon$ and an arbitrary level assignment $a^+$ for the point $p+\epsilon$, we check whether there is a valid folding of the part of the strip between $p-\epsilon$ and $p+\epsilon$ that matches this level assignment. To do so, we check the following four conditions that capture the \emph{noncrossing} conditions defined in Section~\ref{sec:preliminaries}:
\begin{itemize}
\item If a segment $s$ extends to both sides of $p$ without being folded at $p$, then it has the same level on both sides. That is, $a^-(s)=a^+(s)$.
\item For each two input folds at $p$ that connect pairs of segments that overlap $p-\epsilon$, the levels of these pairs of segments are nested or disjoint. That is, if we have a fold connecting segments $s_0$ and $s_1$, and another fold connecting segments $s_2$ and $s_3$, then $[a^-(s_0),a^-(s_1)]$ and $[a^-(s_2),a^-(s_3)]$ are either disjoint intervals or one of these two intervals contains the other.
\item For each two input folds at $p$ that connect pairs of segments that overlap $p+\epsilon$, the levels of these pairs of segments are nested or disjoint. This is a symmetric condition to the previous one, using $a^+$ instead of $a^-$.
\item For each fold at $p$, connecting segments $s_0$ and $s_1$, and for each input segment $s_2$ that crosses $p$ without being folded there, the interval of levels occupied by the fold should not contain the level of $s_2$. That is, if the two segments $s_0$ and $s_1$ extend to the left of $p$, then the interval $[a^-(s_0),a^-(s_1)]$ should not contain $a^-(s_2)$.  If the two segments extend to the right of $p$, then we have the same condition using $a^+$ instead of $a^-$.
\end{itemize}
If the pair $(a^-,a^+)$ passes all these tests, we include $a^+$ in the set of valid level assignments for $p+\epsilon$, which we will then use at the next stopping point of the algorithm.

If, at the end of this process, we reach the rightmost stopping point with a nonempty set of valid level assignments (necessarily consisting of the unique level assignment for the empty set of segments) then a folding of height $k$ exists. The folding itself may be recovered by storing, for each level assignment $a^+$ considered by the algorithm, one of the level assignments $a^-$ such that $a^-\in A$ and $(a^-,a^+)$ passed all the tests above. Then, backtracking through these pointers, from the rightmost stopping point back to the leftmost one, will give a sequence of level assignments such that each consecutive pair is valid, which describes a consistent folding of the entire input strip.

The time for the algorithm is the number of stopping points multiplied by the number of pairs of level assignments for each stopping point and the time to test each pair of level assignments. This is $O(2^{O(k\log k)}n)$, as stated.\qed
\end{proof}

%% file: main.bbl
\begin{thebibliography}{1}

\bibitem{Arkin-Bender-Demaine-Demaine-Mitchell-Sethia-Skiena-2003}
E.~M. Arkin, M.~A. Bender, E.~D. Demaine, M.~L. Demaine, J.~S.~B. Mitchell,
  S.~Sethia, and S.~S. Skiena.
\newblock When can you fold a map?
\newblock {\em Computational Geometry: Theory and Applications}, 29(1):23--46,
  2004.

\bibitem{BannisterEppsteinSimons2012}
M.~J. Bannister, D.~Eppstein, and J.~A. Simons.
\newblock Inapproximability of orthogonal compaction.
\newblock {\em Journal of Graph Algorithms and Applications}, 16:651--673,
  2012.

\bibitem{Bern-Hayes-1996}
M.~Bern and B.~Hayes.
\newblock The complexity of flat origami.
\newblock In {\em Proceedings of the 7th Annual ACM-SIAM Symposium on Discrete
  Algorithms}, pages 175--183, 1996.

\bibitem{DemaineORourke2007}
E.~D. Demaine and J.~O'Rourke.
\newblock {\em {Geometric Folding Algorithms: Linkages, Origami, Polyhedra}}.
\newblock Cambridge University Press, 2007.

\bibitem{Gallivan}
B.~Gallivan.
\newblock Folding paper in half 12 times: {An `Impossible Challenge' Solved and
  Explained}.
\newblock Historical Society of Pomona Valley, 2002.

\bibitem{GJ79}
M.~R. Garey and D.~S. Johnson.
\newblock {\em Computers and Intractability: A Guide to the Theory of
  NP-Completeness}.
\newblock W. H. Freeman \& Co., 1979.

\bibitem{Uehara2010b}
R.~Uehara.
\newblock On stretch minimization problem on unit strip paper.
\newblock In {\em 22nd Canadian Conference on Computational Geometry (CCCG)},
  pages 223--226, 2010.

\bibitem{Uehara2011b}
R.~Uehara.
\newblock Stamp foldings with a given mountain-valley assignment.
\newblock In {\em Origami$^5$: Proceedings of the 5th International Meeting of
  Origami Science, Mathematics, and Education (AK Peters/CRC Press, 2011)},
  pages 585--597, 2011.

\bibitem{UmesatoSaitohUeharaItoOkamoto2013}
T.~Umesato, T.~Saitoh, R.~Uehara, H.~Ito, and Y.~Okamoto.
\newblock The complexity of the stamp folding problem.
\newblock {\em Theoretical Computer Science}, 497:13--19, 2013.

\end{thebibliography}
